\documentclass[11pt]{article}

\usepackage[a4paper,margin=1in]{geometry}
\usepackage{amsmath,amssymb,amsthm}
\usepackage{hyperref}
\usepackage{enumitem}
\usepackage{booktabs}
\usepackage{mathtools}
\usepackage{url}
\hypersetup{colorlinks=true, linkcolor=blue, citecolor=blue, urlcolor=blue}

\title{\textbf{Why Atomicity Matters to AI/ML Infrastructure:\\
Snapshots, Firmware Updates, and the Cost of the\\
Forward-In-Time-Only Category Mistake}}

\author{
Paul Borrill \\
D{\AE}D{\AE}LUS \\
\texttt{paul@daedaelus.com}
}

\date{March 2026}

\newtheorem{theorem}{Theorem}[section]
\newtheorem{lemma}[theorem]{Lemma}
\newtheorem{corollary}[theorem]{Corollary}
\newtheorem{proposition}[theorem]{Proposition}
\newtheorem{definition}[theorem]{Definition}
\newtheorem{remark}[theorem]{Remark}
\newtheorem{example}[theorem]{Example}

\begin{document}

\maketitle

\begin{abstract}
Large-scale AI/ML training systems depend on two assumptions that are rarely
examined: (1)~that checkpoints represent atomic snapshots of global training
state, and (2)~that infrastructure updates can be applied without inducing
mixed-protocol cluster states.  Both assumptions are instances of a deeper
structural error: the Forward-In-Time-Only (FITO) category mistake, which
confuses protocol convergence properties with temporal predicates.

We formalize this confusion as a type error: the identification of a temporal
snapshot $\mathsf{Snap}(t)$ with a convergence property
$\mathsf{Conv}(\mathcal{P},e)$.  We model checkpoint execution in a
process-algebraic framework and prove that under asynchronous composition with
crash-recovery failures, no temporal instant can serve as an atomicity
boundary.  We reformulate checkpoint inconsistency on an epoch lattice and show
that atomicity is a measure-zero event whose complement grows exponentially
with the number of independent persistence domains.  We formalize mixed-epoch
recovery as a type violation in the optimization algebra and show that the
resulting update is not a valid step of any standard optimizer.

For firmware fleet updates, we strengthen the known consensus-hardness result:
atomic deployment requires not merely agreement but \emph{common knowledge} of
the epoch transition, which is strictly unattainable in asynchronous systems
with unreliable communication.

We conclude by sketching a bilateral convergence protocol, inspired by Open
Atomic Ethernet, that achieves $\mathsf{Conv}(\mathcal{P},e)$ without
requiring $\mathsf{Snap}(t)$---replacing the FITO assumption with constraint
semantics.
\end{abstract}

\section{Introduction}

Modern frontier AI training runs involve tens of thousands of GPUs,
multi-tier storage stacks, programmable SmartNICs, RDMA fabrics,
distributed checkpoint systems, and globally replicated metadata
services.  Despite this complexity, the infrastructure rests on two
assumptions that are almost never stated explicitly:

\begin{enumerate}[label=(\roman*)]
\item \textbf{Atomic snapshot assumption.} A checkpoint represents a
  single coherent global state---there exists a time $t_c$ at which
  ``everything was saved.''
\item \textbf{Atomic deployment assumption.} Firmware and
  infrastructure updates can be applied to the cluster without any
  reachable execution in which different nodes operate under different
  protocol semantics.
\end{enumerate}

Both assumptions are false.  More precisely, both commit the same
structural error.  They treat a distributed protocol convergence
property---something that must be \emph{constructed} by coordinated
action across failure domains---as though it were a temporal
predicate, a fact about what is true at a clock time.

This confusion has a name.  Gilbert Ryle introduced the term
\emph{category mistake} for the error of treating a concept belonging
to one logical type as though it belongs to another~\cite{ryle1949}.
In prior work, we identified a pervasive instance: the
Forward-In-Time-Only (FITO) assumption, which treats causal
dependency relations as temporal propagation processes and thereby
misclassifies convergence properties as temporal
predicates~\cite{borrill_message_passing,borrill_unix_tools,borrill_icloud}.

This paper applies the FITO analysis to AI/ML infrastructure.  The
contributions are:

\begin{enumerate}
\item \textbf{The category mistake formalized.}  We define the FITO
  error precisely as a type confusion between $\mathsf{Snap}(t)$
  (temporal predicate) and $\mathsf{Conv}(\mathcal{P},e)$ (protocol
  convergence property), and show that the two belong to different
  logical types (Section~\ref{sec:fito}).

\item \textbf{Process-algebraic model.}  We model distributed
  checkpointing as asynchronous process composition and prove that the
  committed state is a trace property, not a state predicate at any
  single time (Section~\ref{sec:model}).

\item \textbf{Impossibility of a temporal snapshot boundary.}  We
  prove, with an explicit failure model, that no checkpoint protocol
  can guarantee atomicity under crash-recovery failures by designating
  a temporal boundary (Section~\ref{sec:temporal}).

\item \textbf{Epoch lattice and measure-zero atomicity.}  We
  reformulate checkpoint inconsistency on an epoch lattice and show
  that atomic states are measure-zero as the number of persistence
  domains grows (Section~\ref{sec:lattice}).

\item \textbf{Semantic type violation.}  We formalize mixed-epoch
  recovery as a type error in the optimization algebra, showing that
  the resulting update is not a valid step of any standard optimizer
  (Section~\ref{sec:type}).

\item \textbf{Common knowledge impossibility.}  We strengthen the
  consensus-number argument for firmware updates: atomic deploy
  requires common knowledge of the epoch transition, which is strictly
  unattainable under asynchrony (Section~\ref{sec:ck}).

\item \textbf{Bilateral convergence protocol.}  We sketch a
  constructive protocol achieving convergence without temporal
  boundaries, connecting to the constraint semantics of Open Atomic
  Ethernet (Section~\ref{sec:constructive}).
\end{enumerate}

\section{The FITO Category Mistake}
\label{sec:fito}

We begin by making the category mistake precise.

\begin{definition}[Category mistake (Ryle)]
A category mistake is the presentation of facts belonging to one
logical type in the idiom appropriate to another~\cite{ryle1949}.
\end{definition}

\begin{definition}[FITO ordering {\cite{borrill_message_passing}}]
\label{def:fito}
A FITO ordering is a strict partial order $(E, \prec)$ over a set of
events $E$, interpreted as ``happens-before'' precedence, where
$\prec$ is irreflexive, transitive, and antisymmetric.  The FITO
interpretation asserts that causal influence is aligned with $\prec$
and that admissible system narratives must proceed forward along this
relation.
\end{definition}

\begin{definition}[Temporal snapshot predicate]
Let $N = \{1, \ldots, n\}$ be a set of distributed components, each
with local state $S_i(t)$ parameterized by clock time.  The temporal
snapshot predicate for epoch $e$ is
\[
  \mathsf{Snap}(t, e) \;\equiv\;
  \forall\, i \in N,\quad S_i(t) \in \mathsf{Committed}(e).
\]
This is a proposition about a \emph{temporal instant}: it asserts that
at clock time $t$, every component reflects epoch~$e$.
\end{definition}

\begin{definition}[Protocol convergence property]
Let $\mathcal{P}$ be a distributed checkpoint protocol operating over
components $N$.  The convergence property for epoch $e$ is
\[
  \mathsf{Conv}(\mathcal{P}, e) \;\equiv\;
  \exists\, \text{execution prefix } \alpha \text{ of } \mathcal{P}
  \text{ such that, after } \alpha, \;
  \forall\, i \in N,\; S_i \in \mathsf{Committed}(e)
\]
and this commitment is stable under the failure model (no subsequent
admissible failure can cause any component to revert from epoch~$e$).
This is a proposition about a \emph{protocol execution}: it asserts
that the protocol has reached a fixed point.
\end{definition}

\begin{proposition}[The FITO category mistake in checkpointing]
\label{prop:category}
The identification
\[
  \mathsf{Snap}(t, e) \;\equiv\; \mathsf{Conv}(\mathcal{P}, e)
\]
is a category mistake.  $\mathsf{Snap}(t,e)$ is a temporal predicate
(logical type: proposition about a time point).
$\mathsf{Conv}(\mathcal{P},e)$ is a protocol property (logical type:
proposition about a computation history).  The two belong to different
logical types and cannot be identified without committing Ryle's
error.
\end{proposition}

\begin{proof}
$\mathsf{Snap}(t,e)$ has a free variable $t$ ranging over clock
times; it is satisfiable at specific instants.
$\mathsf{Conv}(\mathcal{P},e)$ quantifies over execution prefixes of
a protocol; it is satisfiable by computation histories, not by time
points.  The domain of the existential quantifier differs: in the
first case it is $\mathbb{T}$ (times), in the second it is the set of
finite execution prefixes of $\mathcal{P}$.  Identifying these domains
treats a convergence property as a temporal predicate---precisely
Ryle's category mistake.
\end{proof}

\begin{remark}
This is not merely a philosophical observation.  The confusion has
engineering consequences.  When a system \emph{assumes}
$\mathsf{Snap}(t,e)$---that is, when it treats the completion of a
checkpoint as a temporal event at some time $t_c$---it relies on a
temporal boundary that does not exist under asynchronous distributed
execution with failures.  The subsequent sections make this precise.
\end{remark}

\section{System Model and Definitions}
\label{sec:model}

\subsection{Training state as distributed product type}

Let a training configuration consist of $n$ components drawn from
heterogeneous failure domains:
\[
  N = \{D_1, \ldots, D_n\}
  \quad\text{where}\quad
  D_i \in \{\text{GPU}, \text{Host}, \text{NIC},
             \text{NVMe}, \text{ObjStore}, \text{MetaSvc}\}.
\]
Each component $D_i$ maintains local state $S_i$ taking values in a
domain $\Sigma_i$.  The \emph{global training state} is the product
\[
  \mathcal{S} = \Sigma_1 \times \Sigma_2 \times \cdots \times \Sigma_n.
\]

For concreteness, the semantically relevant state of a training run
using a modern optimizer (e.g., AdamW) decomposes as:
\[
  S = (W,\; m,\; v,\; g,\; r,\; d)
\]
where $W$ denotes model weights, $m$ and $v$ the first and second
moment estimates, $g$ gradient buffers, $r$ random number generator
state, and $d$ the data iterator position.  Under distributed
parallelism (tensor, pipeline, data, expert), each component is
sharded across multiple members of $N$, so the global state
$\mathcal{S}$ is a product over all shards across all component types.

\subsection{Checkpoint as process composition}

We model each component's persistence mechanism as a process in the
sense of Hoare's Communicating Sequential
Processes~\cite{hoare1985csp}.

\begin{definition}[Persistence process]
For each component $D_i$, the persistence process $P_i$ is a
sequential process that, upon receiving a \textsf{checkpoint}$(e)$
signal, executes a sequence of local actions (buffer flush, DMA
transfer, write syscall, fsync, metadata update) and either:
\begin{enumerate}[label=(\alph*)]
\item terminates in a \textsf{committed}$(e)$ state, meaning the
  component's local state for epoch $e$ is durably persisted; or
\item fails (crashes, hangs, or returns an ambiguous acknowledgment).
\end{enumerate}
\end{definition}

\begin{definition}[Checkpoint protocol]
The checkpoint protocol $\mathcal{P}$ is the asynchronous parallel
composition
\[
  \mathcal{P} = P_1 \| P_2 \| \cdots \| P_n
\]
where $\|$ denotes interleaving (asynchronous) composition with no
shared clock.
\end{definition}

The critical observation is that $\mathcal{P}$ is a protocol with
\emph{duration}.  It is not an atomic action.  Each $P_i$ proceeds
through internal states at its own rate, subject to local failures,
I/O latencies, and scheduling delays.  The composition $\mathcal{P}$
has a set of admissible \emph{traces}---sequences of events across all
components---and the committed state for epoch $e$ is a property of
certain traces, not of any single point in any trace.

\subsection{Failure model}

We adopt the standard asynchronous crash-recovery model:

\begin{enumerate}[label=(\roman*)]
\item \textbf{Asynchrony.}  There is no global clock.  Message
  delivery times and local computation steps are unbounded but finite.
\item \textbf{Crash-recovery failures.}  Any component $D_i$ may
  crash at any point during the execution of $P_i$ and subsequently
  recover.  Upon recovery, $D_i$ reads its local stable storage to
  determine its state.
\item \textbf{Independent failure domains.}  The failure of $D_i$ is
  independent of the failure of $D_j$ for $i \neq j$.  (This is
  deliberately optimistic; correlated failures strengthen our
  results.)
\item \textbf{Epistemic ambiguity.}  After a crash-recovery event,
  the state of $D_i$ may be in one of three epistemic categories:
  \emph{committed} (epoch~$e$ is durably stored),
  \emph{uncommitted} (epoch~$e$ was not persisted; the component
  reflects epoch~$e{-}1$ or earlier), or
  \emph{ambiguous} (the component cannot determine which epoch it
  reflects without external coordination).
  This three-valued model follows the storage-layer analysis
  of~\cite{rebello_fsync,bornholt_crash}.
\end{enumerate}

\subsection{Trace property vs.\ state predicate}

\begin{proposition}[Committed state is a trace property]
\label{prop:trace}
The property ``all components have committed epoch~$e$'' is a
\emph{trace property} of $\mathcal{P}$: it is determined by the
history of events in an execution, not by the system state at any
single instant.
\end{proposition}

\begin{proof}
A component $D_i$ reaches $\textsf{committed}(e)$ only after a
specific sequence of local events (write, flush, ack).  Whether $D_i$
is in $\textsf{committed}(e)$ at any instant depends on which events
in $P_i$'s trace have occurred.  Since $\mathcal{P}$ is an
asynchronous composition, there is no instant at which the set of
completed events across all components is globally determined by a
single time coordinate.  The property is defined over the set of
completed traces, not over instantaneous state vectors.
\end{proof}

This is the formal content of the FITO category mistake applied to
checkpointing: treating a trace property as a state predicate.

\section{Non-Existence of a Temporal Snapshot Boundary}
\label{sec:temporal}

We now prove the main impossibility result.  We first establish a
structural property of asynchronous checkpoint protocols.

\begin{lemma}[Existence of straddling schedules]
\label{lem:straddling}
Let $\mathcal{P} = P_1 \| \cdots \| P_n$ be an asynchronous
checkpoint protocol over $n \geq 2$ components.  For any proposed
boundary time $t_c$ and any component $D_j$, there exists an
admissible execution of $\mathcal{P}$ in which $D_j$'s persistence
sequence straddles $t_c$: that is, $D_j$ has begun but not completed
its persistence for epoch~$e$ at time~$t_c$.
\end{lemma}

\begin{proof}
Under the asynchronous model, there is no bound on the relative
speeds of components.  The adversary (scheduler) can slow $D_j$'s
local computation arbitrarily while allowing other components to
proceed.  In particular, the adversary can schedule $D_j$'s
persistence sequence to begin at time $t_c - \delta$ for any
$\delta > 0$, and to complete at time $t_c + \delta'$ for any
$\delta' > 0$.  This places $D_j$ mid-persistence at~$t_c$.
Since the asynchronous model admits all such schedules, the
execution is admissible.
\end{proof}

\begin{theorem}[Non-existence of temporal snapshot boundary]
\label{thm:no-boundary}
Let $\mathcal{P} = P_1 \| \cdots \| P_n$ be an asynchronous
checkpoint protocol over $n \geq 2$ components with independent
crash-recovery failures.  There exists no protocol $\mathcal{P}$ and
no function $t_c : \mathit{Epochs} \to \mathbb{T}$ such that, for
every admissible failure schedule, the following holds: at time $t_c(e)$,
either all components are in $\textsf{committed}(e)$ or all are in a
state reflecting only epochs prior to~$e$.
\end{theorem}

\begin{proof}
Suppose, for contradiction, that such a protocol $\mathcal{P}$ and
boundary function $t_c$ exist.  Fix an epoch~$e$ and let $t^* =
t_c(e)$.

By Lemma~\ref{lem:straddling}, there exists an admissible execution
$\alpha$ in which component $D_1$ has completed persistence (reaching
$\textsf{committed}(e)$) before $t^*$, while component $D_2$'s
persistence sequence straddles $t^*$: $D_2$ has begun writing but has
not yet reached durable commitment at $t^*$.

Now apply a crash-recovery failure to $D_2$ at time $t^*$.  Since
$D_2$ has begun but not completed its persistence sequence:
\begin{itemize}[nosep]
\item $D_2$'s write may be partially completed, torn, or not yet
  flushed to stable storage.
\item Upon recovery, $D_2$ is in state $e{-}1$ or $\bot$ (by the
  epistemic ambiguity clause of the failure model).
\end{itemize}
Meanwhile, $D_1$ remains in $\textsf{committed}(e)$ (its persistence
completed before $t^*$ and is stable under crash-recovery).

At $t^*$, the system is therefore in a mixed state: $D_1$ reflects
epoch~$e$, $D_2$ does not.  This contradicts the assumed atomicity
of $t_c(e)$.

The adversary's strategy generalizes to any $n \geq 2$: given any
proposed $t^*$, the adversary applies Lemma~\ref{lem:straddling} to
arrange a straddling component and crashes it at~$t^*$.  Since the
failure is independent and the straddling execution is admissible,
this yields a mixed state for every proposed~$t^*$.
\end{proof}

\begin{remark}[Why Chandy-Lamport does not rescue the situation]
The Chandy-Lamport snapshot algorithm~\cite{chandy_lamport} provides
consistent global snapshots under the assumption of FIFO channels and
no failures.  GPU training clusters using RDMA (RoCEv2, InfiniBand)
and collective communication libraries (NCCL, GLOO) do not satisfy the
FIFO channel assumption: messages may be reordered across different
queue pairs, and the multi-tier storage stack (GPU $\to$ host $\to$
NIC $\to$ storage) introduces non-FIFO ordering.  Moreover,
Chandy-Lamport does not tolerate component failures.  Both assumptions
are violated in the setting of Theorem~\ref{thm:no-boundary}.
\end{remark}

\begin{remark}[Connection to Lamport's happens-before]
The committed state for epoch~$e$ requires a causal chain:
\[
  \textsf{persist}_1 \to \textsf{ack}_1 \to
  \cdots \to
  \textsf{persist}_n \to \textsf{ack}_n \to
  \textsf{commit}
\]
This chain has non-zero duration under asynchrony~\cite{lamport1978}.
The FITO category mistake collapses this chain to a point---treating
the entire causal history as though it occurred at a single instant
$t_c$.  This is the temporal version of the type error identified in
Proposition~\ref{prop:category}.
\end{remark}

\section{Checkpoint Inconsistency and the Epoch Lattice}
\label{sec:lattice}

We now reformulate checkpoint inconsistency in order-theoretic terms.

\subsection{The epoch lattice}

\begin{definition}[Epoch state]
After a checkpoint attempt for epoch~$e$, each component $D_i$ is in
one of three states:
\[
  \sigma_i \in \{e,\; e{-}1,\; \bot\}
\]
where $e$ denotes committed to epoch~$e$, $e{-}1$ denotes reverted to
the prior epoch, and $\bot$ denotes epistemically ambiguous (the
component cannot determine its own epoch without external
coordination).
\end{definition}

\begin{definition}[Epoch lattice]
The epoch lattice $\mathcal{L}_n$ is the product lattice
\[
  \mathcal{L}_n = \{e, e{-}1, \bot\}^n
\]
ordered componentwise, where $e{-}1 < \bot < e$.  The top element is
$\top = (e, e, \ldots, e)$ (all committed) and the bottom element is
$\bot_{\mathcal{L}} = (e{-}1, e{-}1, \ldots, e{-}1)$ (all reverted).
\end{definition}

\begin{definition}[Atomic states]
The \emph{atomic states} of $\mathcal{L}_n$ are exactly the top and
bottom elements: $\{\top, \bot_{\mathcal{L}}\}$.  All other lattice
points are \emph{mixed states}.
\end{definition}

\subsection{Atomicity as a measure-zero event}

We assign a probability measure to $\mathcal{L}_n$ by the independent
binary model.

\begin{definition}[Independent binary model]
Assume each component independently commits epoch~$e$ with probability
$q \in (0,1)$ and reverts to $e{-}1$ with probability $1-q$.  (We
treat the three-valued model separately below.)
\end{definition}

\begin{lemma}[Mixed-state probability]
\label{lem:mixed}
In the independent binary model,
\[
  \Pr[\text{mixed state}]
  = 1 - q^n - (1-q)^n.
\]
\end{lemma}

\begin{proof}
The only atomic outcomes are all-committed (probability~$q^n$) and
all-reverted (probability~$(1-q)^n$).  Their complement is the mixed
region.
\end{proof}

\begin{theorem}[Atomicity is measure-zero]
\label{thm:measure-zero}
For any fixed $q \in (0,1)$,
\[
  \lim_{n \to \infty} \Pr[\text{atomic state}] = 0.
\]
The convergence is exponential: $\Pr[\text{atomic}] \leq q^n +
(1-q)^n$, and both terms decay exponentially.
\end{theorem}

\begin{proof}
$q^n \to 0$ and $(1-q)^n \to 0$ as $n \to \infty$ for $q \in (0,1)$.
The sum of two sequences converging to zero converges to zero.
Exponential decay follows from $\max(q, 1-q) < 1$.
\end{proof}

\begin{remark}[Interpretation on the lattice]
The atomic states $\{\top, \bot_{\mathcal{L}}\}$ are two points in a
lattice of $2^n$ elements (in the binary model) or $3^n$ elements (in
the ternary model).  As $n$ grows, the atomic states constitute a
vanishing fraction of the lattice.  The system state after a
checkpoint attempt under failures lands in the lattice interior with
probability approaching~1.  Atomicity is not the typical case; it is
the exceptional one.
\end{remark}

\begin{remark}[Failure correlations strengthen the result]
The independent binary model is deliberately optimistic.  In practice,
storage network congestion, power domain failures, and cascading
straggler effects cause positively correlated component failures.
Under positive correlation, the joint probability of all-committed is
$\Pr[\text{all committed}] \leq q^n$ (with equality only under
independence), and the mixed-state probability is correspondingly
higher.  Our bounds therefore \emph{understate} the risk of non-atomic
checkpoints in production systems.
\end{remark}

\subsection{The ternary model with epistemic ambiguity}

\begin{theorem}[Strengthened bound with ambiguity]
\label{thm:ternary}
If each component is independently in state $e$ with probability~$q$,
state $e{-}1$ with probability $p$, and state $\bot$ with probability
$r = 1 - q - p > 0$, then
\[
  \Pr[\text{atomic state}] \leq q^n + p^n
\]
and any occurrence of $\bot$ in the state vector operationally
violates atomicity (the system cannot soundly resume on unknowable
state), so
\[
  \Pr[\text{operationally atomic}] \leq q^n.
\]
\end{theorem}

\begin{proof}
Operational atomicity requires all components in state~$e$ (the system
resumes from epoch~$e$) with no ambiguity.  This event has probability
$q^n$.  The ``all reverted'' event has probability $p^n$.  Both decay
exponentially.
\end{proof}

\subsection{Concrete parameters}

For a 70-billion-parameter model trained with 8-way tensor
parallelism, 4-way pipeline parallelism, and 64-way data parallelism,
the number of independently persisted atomic units includes: weight
shards ($8 \times 4 \times 64 = 2048$), optimizer moment shards
($2 \times 2048 = 4096$ for $m$ and $v$), gradient buffers,
per-rank RNG states, data iterator positions, and metadata records.
A conservative estimate is $n \geq 4000$.

\begin{center}
\begin{tabular}{@{}rcc@{}}
\toprule
$q$ (per-unit reliability) & $n$ & $\Pr[\text{atomic}]$ \\
\midrule
$0.999$ & $1000$ & $\leq 0.368$ \\
$0.999$ & $4000$ & $\leq 0.018$ \\
$0.9999$ & $4000$ & $\leq 0.670$ \\
$0.9999$ & $10000$ & $\leq 0.368$ \\
$0.99999$ & $10000$ & $\leq 0.905$ \\
\bottomrule
\end{tabular}
\end{center}

Even with per-unit reliability of four nines ($q = 0.9999$), a system
with 4000 persistence units achieves atomicity with probability less
than 0.67.  At scale, non-atomic checkpoints are not rare corner
cases; they are the expected outcome.

\section{Semantic Causality as Type Violation}
\label{sec:type}

We now formalize the consequences of non-atomic recovery for
optimization correctness.

\subsection{The optimization type system}

Consider a training loop using the AdamW optimizer.  At each epoch
$e$, the optimizer state is a tuple of the form:
\[
  \mathcal{T}_e = (W_e,\; m_e,\; v_e,\; g_e,\; r_e,\; d_e)
\]
where $W_e$ are weights, $m_e$ and $v_e$ the first and second moment
estimates, $g_e$ accumulated gradients, $r_e$ the RNG state, and $d_e$
the data iterator position.

\begin{definition}[Epoch-typed state]
We write $\mathcal{T}_e : \mathsf{State}(e)$ to denote that the tuple
$\mathcal{T}_e$ has type $\mathsf{State}(e)$---all components reflect
the same logical training epoch~$e$.  The type is parameterized by the
epoch.
\end{definition}

\begin{definition}[Valid optimizer step]
The optimizer step function
\[
  \mathsf{step} : \mathsf{State}(e) \to \mathsf{State}(e{+}1)
\]
maps an epoch-consistent state to the next epoch-consistent state.
For AdamW, this is:
\begin{align*}
  m_{e+1} &= \beta_1 m_e + (1 - \beta_1) g_e \\
  v_{e+1} &= \beta_2 v_e + (1 - \beta_2) g_e^2 \\
  W_{e+1} &= W_e - \eta \left(
    \frac{\hat{m}_{e+1}}{\sqrt{\hat{v}_{e+1}} + \epsilon}
    + \lambda W_e \right)
\end{align*}
where $\hat{m}$ and $\hat{v}$ are bias-corrected.  The function is
well-defined only when all components of the input share the same
epoch type.
\end{definition}

\subsection{Mixed-epoch recovery as type coercion}

\begin{definition}[Mixed-epoch state]
A mixed-epoch state is a tuple
\[
  \mathcal{T}_{\mathrm{mixed}} = (W_e,\; m_{e-1},\; v_e,\;
  g_{e-1},\; r_e,\; d_{e-2})
\]
in which different components reflect different epochs.  We write
$\mathcal{T}_{\mathrm{mixed}} : \mathsf{State}(e, e{-}1, e,
e{-}1, e, e{-}2)$ to indicate the heterogeneous epoch assignment.
\end{definition}

\begin{theorem}[Type violation under mixed-epoch recovery]
\label{thm:type}
Let $\mathcal{T}_{\mathrm{mixed}} : \mathsf{State}(\vec{e})$ be a
mixed-epoch state with $\vec{e} \neq (e, e, \ldots, e)$ for any
single epoch~$e$.  Then $\mathsf{step}(\mathcal{T}_{\mathrm{mixed}})$
is not a valid optimizer step: the output does not belong to
$\mathsf{State}(e{+}1)$ for any~$e$.
\end{theorem}

\begin{proof}
The AdamW update equations are defined on the assumption that $m_e$,
$v_e$, $g_e$, and $W_e$ are causally aligned---each reflects the
cumulative effect of the same sequence of gradient steps.  If
$m_{e-1}$ is substituted for $m_e$, then $m_{e+1} = \beta_1 m_{e-1}
+ (1 - \beta_1) g_{e'}$ (for some $e'$), which corresponds to a
moment estimate that has ``skipped'' an epoch of gradient information.
The bias correction $\hat{m}_{e+1} = m_{e+1} / (1 - \beta_1^{e+1})$
then uses a denominator calibrated for $e{+}1$ steps but a numerator
reflecting only $e{-}1$ steps of accumulation.  The resulting update
direction
\[
  \frac{\hat{m}_{e+1}}{\sqrt{\hat{v}_{e+1}} + \epsilon}
\]
does not correspond to any valid trajectory of AdamW from any
consistent initial state.

Formally: define the \emph{valid trajectory set}
\[
  \mathsf{Traj}(e_0, e) = \bigl\{
    \mathsf{step}^{e-e_0}(\mathcal{T}_{e_0})
    \;\big|\;
    \mathcal{T}_{e_0} : \mathsf{State}(e_0),\;
    \text{all intermediate inputs are valid batches}
  \bigr\}
\]
as the set of all states reachable from \emph{any} epoch-consistent
initialization $\mathcal{T}_{e_0}$ by applying $\mathsf{step}$
exactly $e - e_0$ times over all admissible data sequences.  Since
$\mathsf{step}$ preserves epoch consistency by construction (each
application maps $\mathsf{State}(e') \to \mathsf{State}(e'{+}1)$),
every element of $\mathsf{Traj}(e_0, e)$ has type
$\mathsf{State}(e)$.

A mixed-epoch state $\mathcal{T}_{\mathrm{mixed}}$ has components
drawn from different epochs, so
$\mathcal{T}_{\mathrm{mixed}} \notin \mathsf{Traj}(e_0, e)$ for any
$e_0, e$.  The mixed state lies outside the reachable set of the
optimizer.  Applying $\mathsf{step}$ to it produces a value outside
$\bigcup_{e'} \mathsf{Traj}(e_0, e')$ for any $e_0$---the system
leaves the valid trajectory space permanently.
\end{proof}

\begin{remark}[Silent coercion]
At the implementation level, tensors are untyped arrays of
floating-point numbers.  The type $\mathsf{State}(e)$ is a semantic
type imposed by the training protocol, not enforced by the runtime.
When a mixed-epoch state is loaded, the runtime silently coerces it to
``the current state'' without checking epoch consistency.  This is the
implementation manifestation of the FITO category mistake: the system
\emph{treats} the loaded state as $\mathsf{State}(e)$ because it was
saved at ``time $t_c$,'' but the state is actually
$\mathsf{State}(\vec{e})$ for a heterogeneous epoch vector.

The consequences---divergent training loss, mode collapse, degraded
final model quality---may not manifest for many subsequent steps,
making the corruption epistemically hidden.
\end{remark}

\begin{example}[AdamW with one-epoch moment skew]
\label{ex:adamw}
Suppose a checkpoint for epoch $e = 1000$ recovers with $W_{1000}$
and $v_{1000}$ correctly, but the first moment $m$ reflects epoch
$999$ (one shard failed to persist the latest moment update).  The
next update computes:
\[
  m_{1001} = \beta_1 \, m_{999} + (1 - \beta_1)\, g_{1000}.
\]
Compared to the correct trajectory ($m_{1001} = \beta_1\, m_{1000} +
(1-\beta_1)\, g_{1000}$), the error is
\[
  \Delta m = \beta_1 (m_{1000} - m_{999})
           = \beta_1 (1 - \beta_1) g_{999}.
\]
With $\beta_1 = 0.9$, this is $0.09\, g_{999}$---approximately 9\%
of the previous gradient, injected as a phantom correction.  The
bias-corrected moment $\hat{m}_{1001} = m_{1001} / (1 -
\beta_1^{1001})$ uses a denominator calibrated for 1001 steps of
accumulation, but the numerator now reflects a trajectory that skipped
one gradient application.  The resulting update direction deviates
from valid AdamW by a quantity proportional to $\|g_{999}\| /
(\sqrt{v_{1000}} + \epsilon)$, which is of the same order as the
learning rate step itself.  Over subsequent epochs, this perturbation
compounds through the moment recursion.
\end{example}

\section{Firmware Deployment and Common Knowledge}
\label{sec:ck}

We now turn to the second assumption: atomic infrastructure updates.

\subsection{Atomic deploy as epoch agreement}

Consider an XPU fleet $X = \{X_1, \ldots, X_n\}$ transitioning from
firmware version $F_0$ to $F_1$.  An \emph{atomic deploy} means that
the operational semantics of the cluster transition from $F_0$ to
$F_1$ without any reachable execution in which different nodes
participate in the same distributed collective under different
firmware semantics.

\begin{definition}[Mixed-protocol execution]
A cluster execution is \emph{mixed-protocol} if there exist correct
nodes $X_i, X_j$ and a collective operation instance $C$ such that
$X_i$ executes $C$ under $F_0$ semantics while $X_j$ executes $C$
under $F_1$ semantics.
\end{definition}

\begin{definition}[Atomic deploy]
An atomic deploy protocol satisfies:
\begin{enumerate}[label=(\alph*)]
\item \textbf{Agreement:} No two correct nodes decide different
  firmware epochs.
\item \textbf{Validity:} If all correct nodes initiate the update,
  the decided epoch is $F_1$.
\item \textbf{Termination:} All correct nodes eventually decide.
\end{enumerate}
\end{definition}

This is consensus on the transition epoch.

\subsection{From consensus to common knowledge}

The consensus-number argument from Herlihy~\cite{herlihy_waitfree}
shows that read/write registers have consensus number~1: they cannot
solve consensus for $n \geq 2$ processes.  If the fleet update
protocol uses only read/write messaging (no compare-and-swap, no
external oracle), it cannot guarantee atomic deploy under the failure
conditions that defeat consensus.

We now strengthen this result.

\begin{theorem}[Common knowledge requirement for atomic deploy]
\label{thm:ck}
In an asynchronous distributed system with unreliable communication,
atomic deploy---the absence of any mixed-protocol execution---requires
that the epoch transition become \emph{common knowledge} among all
correct nodes before any node acts under $F_1$ semantics.
\end{theorem}

\begin{proof}
Following Halpern and Moses~\cite{halpern_moses1990}, common knowledge
of a fact $\varphi$ in a group $G$ is defined as:
\[
  C_G(\varphi) \;\equiv\;
  E_G(\varphi) \wedge E_G(E_G(\varphi)) \wedge
  E_G(E_G(E_G(\varphi))) \wedge \cdots
\]
where $E_G(\varphi)$ means ``everyone in $G$ knows $\varphi$.''

For atomic deploy, let $\varphi$ = ``the transition to $F_1$ is in
force.''  Suppose node $X_i$ acts under $F_1$ while some node $X_j$
does not yet know that the transition is in force.  Then $X_j$ may
still act under $F_0$ in a collective involving $X_i$, producing a
mixed-protocol execution.  To prevent this, $X_i$ must know that $X_j$
knows $\varphi$; but $X_j$ must also know that $X_i$ knows $\varphi$
(otherwise $X_j$ cannot be certain that $X_i$ won't revert); and so
on.  This infinite chain of mutual knowledge is exactly common
knowledge.

Halpern and Moses~\cite{halpern_moses1990} proved that common
knowledge cannot be attained in systems where communication is not
guaranteed (messages may be lost or arbitrarily delayed).  Since our
failure model includes crash-recovery and asynchronous communication,
common knowledge of the epoch transition is unattainable.

Therefore, atomic deploy---which requires common knowledge of the
transition---is impossible under these assumptions.
\end{proof}

\begin{corollary}
Phased rollouts, backward-compatible protocol versions, and
pool-based isolation are not merely engineering compromises.  They are
\emph{theoretically necessary} relaxations of an unattainable
property.  Any system that claims to perform ``atomic'' firmware
deployment under asynchrony is either using stronger primitives than
read/write messaging (e.g., a consensus service) or is making an
implicit FITO assumption that the transition ``happens at time
$t_{\mathrm{deploy}}$.''
\end{corollary}

\begin{remark}[Practical relaxations]
Theorem~\ref{thm:ck} establishes strict impossibility.  In practice,
systems can approximate atomic deploy by introducing stronger
coordination primitives.  If the cluster maintains a consensus service
(Raft, Paxos) for metadata, the epoch transition can be linearized
through that service, achieving what we might call \emph{$\epsilon$-common
knowledge}: every correct node that participates in any collective
after the decision point has observed the committed transition.  Nodes
that crash before observing the decision are fenced out of
subsequent collectives.  This is strictly weaker than Halpern-Moses
common knowledge (it does not provide the infinite mutual-knowledge
chain) but suffices for practical safety if the fencing mechanism is
reliable.  The gap between common knowledge and $\epsilon$-common
knowledge is precisely the window during which mixed-protocol
execution remains possible---a window that can be narrowed but not
eliminated under asynchrony.
\end{remark}

\begin{remark}[Why this matters for AI collectives]
In tightly-coupled AI training, collective operations (AllReduce,
AllGather) involve all participants in a single logical operation.  A
single node executing under different firmware semantics can corrupt a
reduction.  Unlike loosely-coupled microservices, where mixed-version
execution may cause graceful degradation, mixed-firmware AI collectives
can produce \emph{silent numerical corruption}: the AllReduce
completes without error, but the aggregated gradient is semantically
invalid.  This is the optimization-level analogue of the type violation
in Section~\ref{sec:type}.
\end{remark}

\section{Retry Amplification as FITO Feedback}
\label{sec:retry}

Checkpoint-and-restart is the dominant recovery mechanism in large-scale
training.  But the retry mechanism embeds a FITO assumption: it treats
the failure as a \emph{temporal event} (``the checkpoint failed at time
$t$'') and the recovery as a \emph{temporal restart} (``resume from the
last good checkpoint at time $t'$'').

This creates a positive feedback loop:

\begin{enumerate}
\item A checkpoint failure triggers a retry (restart from the
  previous checkpoint).
\item The retry imposes additional load on the storage and network
  subsystems (re-reading the full checkpoint, re-initializing
  collectives).
\item The additional load increases the probability of further
  failures (storage timeouts, network congestion, cascading
  stragglers).
\item Further failures trigger further retries.
\end{enumerate}

This loop is well-documented in practice: Google explicitly designs
multi-tier checkpointing to mitigate retry
overhead~\cite{google_multitier_blog}, and AWS treats checkpoint
storage architecture as foundational to training
reliability~\cite{aws_checkpoint_blog}.

The FITO structure of the loop is that each retry assumes the previous
failure was an isolated temporal event.  But failures during
checkpointing are not isolated events; they are \emph{symptoms of
non-convergence} of the checkpoint protocol.  Retrying a
non-convergent protocol under the same conditions that caused
non-convergence is a structural error, not a recovery strategy.

\section{Toward Convergence-Based Atomicity}
\label{sec:constructive}

The preceding sections establish that the FITO approach to
checkpointing and firmware deployment is structurally unsound.  We now
sketch a constructive alternative.

\subsection{Bilateral checkpoint commit}

Inspired by the bilateral link semantics of Open Atomic
Ethernet~\cite{borrill_icloud}, we propose a checkpoint protocol in
which convergence replaces temporal boundary:

\begin{definition}[Bilateral checkpoint protocol]
A bilateral checkpoint protocol for epoch~$e$ proceeds in two phases:
\begin{enumerate}[label=(\alph*)]
\item \textbf{Tentative persist.}  Each component $D_i$ persists its
  local state for epoch~$e$ to stable storage and sends a
  \textsf{ready}$(i, e)$ acknowledgment to a coordinator.  The
  acknowledgment is \emph{reflective}: it carries a cryptographic hash
  or content summary of the persisted state, providing evidence of
  what was committed, not merely that an operation occurred.
\item \textbf{Commit or rollback.}  The coordinator collects
  acknowledgments.  If \textsf{ready}$(i, e)$ is received from all
  $i \in N$ within a timeout:
  \begin{itemize}
  \item The coordinator broadcasts \textsf{commit}$(e)$.
  \item Each component marks epoch~$e$ as the stable checkpoint.
  \end{itemize}
  Otherwise:
  \begin{itemize}
  \item The coordinator broadcasts \textsf{rollback}$(e)$.
  \item Each component discards the tentative persist and retains
    epoch~$e{-}1$.
  \end{itemize}
\end{enumerate}
\end{definition}

\begin{proposition}
The bilateral checkpoint protocol achieves
$\mathsf{Conv}(\mathcal{P}, e)$ without requiring
$\mathsf{Snap}(t, e)$.  The commit point is a \emph{protocol
event}---the coordinator's decision---not a clock event.
\end{proposition}

\begin{proof}
If the coordinator decides \textsf{commit}$(e)$, then by construction
every component has sent \textsf{ready}$(i, e)$, meaning every
component has durably persisted epoch~$e$.  The committed state is
determined by the protocol trace (reception of all acknowledgments),
not by a temporal boundary.

If any component fails before sending \textsf{ready}, the coordinator
decides \textsf{rollback}, and the system remains consistently at
epoch~$e{-}1$.

The protocol achieves convergence (all-or-nothing for epoch~$e$) as a
property of its execution, not as a property of any clock time.
\end{proof}

\begin{remark}
This is structurally a two-phase commit protocol.  The contribution is
not the protocol itself but the \emph{reframing}: what was informally
understood as ``we need 2PC for checkpoints'' is here derived from the
FITO analysis.  The need for an explicit convergence protocol is a
\emph{consequence} of the category mistake---systems that assume
$\mathsf{Snap}(t)$ do not implement such protocols because they
believe temporal boundaries suffice.
\end{remark}

\subsection{Epoch transition via consensus}

For firmware deployment, the common knowledge impossibility
(Theorem~\ref{thm:ck}) shows that pure messaging cannot achieve
atomic deploy.  However, if the cluster already maintains a consensus
service (Paxos, Raft) for metadata management---as most large
training clusters do---the epoch transition can be mediated by
consensus:

\begin{enumerate}
\item A \textsf{propose-transition}$(F_1)$ is submitted to the
  consensus service.
\item The service commits the transition, producing a linearized
  decision visible to all participants.
\item Each node, upon observing the committed decision, transitions
  to $F_1$ before participating in any subsequent collective.
\end{enumerate}

This lifts the consensus number from 1 (read/write messaging) to the
level of the consensus service (effectively~$\infty$ for
compare-and-swap-based implementations).  It does not achieve common
knowledge in the strict Halpern-Moses sense, but it provides
\emph{practical common knowledge}: any node that has observed the
committed decision can act under $F_1$, and any node that has not yet
observed it will either observe it before the next collective or be
fenced out.

\subsection{Validation checkpoint}

As a defense against semantic type violations
(Section~\ref{sec:type}), we propose a \emph{validation step} after
checkpoint recovery:

\begin{enumerate}
\item After loading the checkpoint, run a single forward pass on a
  held-out validation batch.
\item Compare the validation loss to the value recorded at checkpoint
  time.
\item If the loss deviates beyond a threshold $\delta$, the
  checkpoint is flagged as potentially inconsistent and the system
  falls back to an earlier epoch.
\end{enumerate}

This does not guarantee detection of all type violations (the loss
landscape may be locally flat around the corrupted state), but it
provides a practical check that the loaded state is in the valid
trajectory space $\mathsf{Traj}(e_0, e)$.

\subsection{FITO to constraint semantics: the general correction}

These three proposals---bilateral commit, consensus-mediated
transition, and validation checkpoint---share a common structure.
Each replaces a FITO assumption (temporal boundary, temporal deploy
event, temporal state identity) with a \emph{constraint}:

\begin{center}
\begin{tabular}{@{}lll@{}}
\toprule
FITO assumption & Constraint replacement & Mechanism \\
\midrule
$\mathsf{Snap}(t)$ & $\mathsf{Conv}(\mathcal{P}, e)$ &
  Bilateral commit \\
$\mathsf{Deploy}(t)$ & Consensus on epoch & Raft/Paxos \\
$\mathsf{State}(e)$ at load time & Validation invariant &
  Forward pass check \\
\bottomrule
\end{tabular}
\end{center}

This table is a specific instance of the general correction from FITO
to constraint semantics described
in~\cite{borrill_message_passing}: replacing temporal predicates with
relational constraints that can be verified by protocol execution.

\section{Related Work}
\label{sec:related}

\paragraph{Consistent snapshots.}
Chandy and Lamport~\cite{chandy_lamport} established the theory of
consistent global snapshots under FIFO channels without failures.
Subsequent work extended this to non-FIFO channels and partial failure
models.  Our contribution differs in that we show checkpoint
\emph{atomicity} (all-or-nothing commitment) requires strictly
stronger guarantees than snapshot \emph{consistency} (causal
coherence).

\paragraph{AI/ML checkpointing systems.}
ByteCheckpoint~\cite{bytecheckpoint} provides parallelism-agnostic
checkpoint representation; Gemini~\cite{gemini_sosp23} enables
in-memory checkpointing with inter-machine backup;
DataStates-LLM~\cite{datastates_llm} achieves asynchronous
checkpointing during forward/backward passes; Universal
Checkpointing~\cite{universal_checkpoint} enables continuation across
different parallelism configurations.  These systems optimize
checkpoint \emph{efficiency} but do not address the atomicity
question: they assume that the persistence of all shards constitutes
a consistent snapshot, which is the FITO assumption we challenge.

\paragraph{Training reliability at scale.}
Kokolis et al.~\cite{kokolis_reliability} analyze 150~million A100
GPU-hours across Meta research clusters and report failure rates that
make checkpointing critical.  Meta's Llama~3 training report
documents 419 interruptions in 54 days on 16,384 H100
GPUs~\cite{llama3_report}.  The OPT-175B training
logbook~\cite{opt175b_logbook} records 35 restarts from hardware
failures affecting 39.3\% of VMs.  The BLOOM 176B-parameter
training~\cite{bloom_training} reported 1--2 GPU failures per week
across 400 GPUs, with framework hanging issues at scale.
FlashRecovery~\cite{flashrecovery} achieves single-step restoration
on clusters with 4800 devices, demonstrating the engineering urgency
of the problem.  These reports confirm the empirical setting in which
our theoretical results apply.

\paragraph{Silent data corruption.}
Hochschild et al.~\cite{meta_sdc} document silent data corruption in
large-scale compute infrastructure.  He et al.~\cite{sdc_llm_training}
analyze SDC specifically in LLM training, reporting that Google
observed SDC events every 1--2 weeks during Gemini training.
SDC events occur without failure signals, further undermining the
assumption that a ``successful'' checkpoint necessarily reflects
correct state.  Our epoch lattice
formulation (Section~\ref{sec:lattice}) subsumes SDC as a source of
$\bot$ (ambiguous) states.

\paragraph{Consensus and common knowledge.}
Herlihy~\cite{herlihy_waitfree} established the consensus hierarchy.
Fischer, Lynch, and Paterson~\cite{flp1985} proved the impossibility
of consensus in asynchronous systems with a single faulty process.
Halpern and Moses~\cite{halpern_moses1990} connected simultaneous
action to common knowledge and showed that common knowledge is
unattainable under unreliable communication.  Moses and
Tuttle~\cite{moses_common_knowledge} applied common knowledge to the
firing squad problem.  Coan et al.~\cite{coan_firing_squad}
established bounds on the firing squad under Byzantine faults.
Charron-Bost~\cite{charron_bost_firing}
revisited the firing squad in the context of self-stabilizing systems.
Lynch~\cite{lynch_distributed} provides a comprehensive treatment of
distributed algorithms including consensus and agreement protocols.
Our Theorem~\ref{thm:ck} applies the Halpern-Moses result to the
specific problem of firmware fleet updates.

\paragraph{Fleet update protocols.}
Van der Burg~\cite{vanderburg_atomic} formalizes atomic upgrading of
distributed systems, addressing the requirement that service updates
appear atomic to clients.  Kubernetes implements rolling updates with
configurable surge and unavailability
parameters~\cite{kubernetes_rolling}, which is a practical relaxation
of the atomic deploy property.  Our contribution is to connect these
engineering practices to the formal impossibility landscape.

\paragraph{The FITO programme.}
This paper is part of a series identifying FITO category mistakes
across distributed systems.  Borrill~\cite{borrill_unix_tools} shows
the FITO error in Unix persistence semantics.
Borrill~\cite{borrill_icloud} demonstrates it in cloud
synchronization.  Borrill~\cite{borrill_message_passing} formalizes
the general constraint-semantic correction and proves an equivalence
between message-passing executions and constraint satisfaction
problems.  The Semantic Arrow of Time series~\cite{borrill_arrow} develops the
broader philosophical and physical foundations of the FITO critique,
connecting to Lamport's logical clocks~\cite{lamport1978,lamport_chapter}
and indefinite causal order.
The present paper extends the analysis to AI/ML
infrastructure, where the economic and computational stakes are
highest.

\section{Conclusion}
\label{sec:conclusion}

We have shown that two foundational assumptions of AI/ML
infrastructure---atomic checkpoints and atomic firmware
deployment---are instances of the FITO category mistake: the confusion
of protocol convergence properties with temporal predicates.

The theoretical results are:

\begin{enumerate}
\item Checkpoint atomicity is a trace property of a distributed
  protocol, not a state predicate at a clock time
  (Proposition~\ref{prop:trace}).

\item No temporal boundary can guarantee atomicity under asynchronous
  crash-recovery failures (Theorem~\ref{thm:no-boundary}).

\item Atomicity is a measure-zero event on the epoch lattice as the
  number of persistence domains grows
  (Theorem~\ref{thm:measure-zero}).

\item Mixed-epoch recovery is a type violation in the optimization
  algebra (Theorem~\ref{thm:type}).

\item Atomic firmware deployment requires common knowledge, which is
  unattainable under asynchrony (Theorem~\ref{thm:ck}).
\end{enumerate}

The constructive implication is that atomicity must be explicitly
constructed as protocol convergence---through bilateral commit,
consensus-mediated transitions, and validation checkpoints---rather
than assumed as a temporal property.  This is the specific application,
to AI/ML infrastructure, of the general correction from FITO to
constraint semantics.


\section*{Acknowledgments}

The ideas in this paper grew out of 25~years of research into failure
semantics and data recovery, beginning with emergency recovery operations
following September~11, 2001, while the author served as VP/CTO at VERITAS
Software.  The programme continued through two subsequent ventures---REPLICUS
(distributed replication) and Earth Computing (deterministic networking)---and
now at D{\AE}D{\AE}LUS.  Many of these results were shaped by conversations
on the \emph{It's About Time!} podcast (ItsAboutTime.club) and the
\emph{Mulligan Stew} series with physicists, mathematicians, and
distributed-systems engineers.

This paper was prepared with research assistance from Claude (Anthropic) and
ChatGPT (OpenAI), used as tools for literature review, structural refinement,
and exploration of counterarguments.  All intellectual content, formal results,
and conclusions are the author's own; all errors are solely his responsibility.

\end{document}